\documentclass[letterpaper, 12pt,twoside]{article}
\usepackage{latexsym}
\usepackage{amsmath}
\usepackage{mathrsfs}
\usepackage{amssymb}
\usepackage[latin1]{inputenc}
\usepackage{graphicx}
\usepackage{fancyhdr}
\usepackage{color}
\usepackage{amsthm}

\newcommand{\labell}[1]{\label{#1}\qquad_{#1}} 
\newcommand{\bbibitem}[1]{\bibitem{#1}\marginpar{#1}}
\newcommand{\llabel}[1]{\label{#1}\marginpar{#1}}
\newtheorem{lem}{Lemma}
\newtheorem{thm}{Theorem}

\def\Label#1{\label{#1}%
  \smash{\hbox to0pt{\raise1ex\hbox{\tiny[#1]}\hss}}}
\def\noLabels{\let\Label=\label}
\def\nolabells{\let\labell=\label}
\def\nollabels{\let\llabel=\label}
\def\nobbibitem{\let\bbibitem=\bibitem}

\title{{\bf Consistently melting crystals}}

\author{Klaus Larjo\footnote{larjo@phas.ubc.ca}
\\[1mm]
\small \sl Department of Physics and Astronomy, University of British Columbia, \\
[-1.5mm]
\small \sl  Vancouver, B.C. V6T 1Z1, Canada\\
}
\date{}

\begin{document}

\noLabels \nollabels \setlength{\baselineskip}{16pt} \nobbibitem
\begin{titlepage}

\maketitle

\begin{abstract}
Recently Ooguri and Yamazaki proposed a statistical model of melting crystals to count BPS bound states of certain D-brane
configurations on toric Calabi--Yau manifolds [arXiv:0811.2801]. This construction relied on a set of consistency conditions on
the corresponding brane tiling, and in this note I show that these conditions are satisfied for any physical brane tiling; they
follow from the conformality of the low energy field theory on the D-branes. As a byproduct I also provide a simple direct proof
that any physical brane tiling has a perfect matching.
\end{abstract}
\thispagestyle{empty} \setcounter{page}{0}
\end{titlepage}

\section{Introduction}
\llabel{intro} Understanding the gauge theory on a stack of D-branes probing a Calabi--Yau singularity is an important problem
with many applications.  While the case of a general Calabi--Yau manifold remains poorly understood, in the past few years
considerable progress has been made in the case of toric Calabi--Yaus. For these, the low energy field theory on the D-branes is
given by a quiver gauge theory, and can be analyzed using powerful brane tiling techniques initiated and developed in
\cite{dimer1,dimer2,dimer3,rhombus,amoeba}.

A different way of understanding such configurations of D-branes is through the derived category approach proposed in
\cite{kont,douglas1,asp-zero}; see \cite{asp-review} for an excellent review. While working with the derived category of coherent
sheaves, $D(\textrm{Coh }X)$, is again difficult for a general Calabi--Yau $X$, the situation again simplifies when $X$ is toric
and tools like exceptional collections can be used to analyze the system \cite{asp-exp,her-exp,han-exp,asp-tor}.

In \cite{ooguri}, based on the mathematical work of Mozgovoy and Reineke \cite{reineke}, Ooguri and Yamazaki proposed a
statistical model of crystal melting that counts bound states of D0 and D2 branes in the background of a single D6 brane wrapping
the whole toric Calabi--Yau. This model utilizes both approaches mentioned above, and can be used for instance to compute
Donaldson--Thomas invariants for an arbitrary toric Calabi--Yau. This crystal is built upon the planar quiver describing the
theory, and the construction relies on a set of consistency conditions on the corresponding brane tiling. In this note I show
that these consistency conditions are indeed satisfied for any quiver theory arising from a configuration of D-branes; they are
shown to follow from the conformality of the low energy field theory.

\section{Toric quiver theories and brane tilings}
\llabel{review}  In this section I briefly review quiver gauge theories arising as low energy field theories on the world volume
of a stack of D-branes probing the conical singularity of a toric Calabi--Yau; for excellent extended reviews on this topic the
reader is referred to \cite{kennaway,yamazaki}.

\paragraph{Quivers:} Quiver gauge theories are specified by a set $Q_0$ of gauge groups SU($N_i$), and a set $Q_1$ of matter fields transforming in the bifundamental
representation $(N_i,\bar{N}_j)$ under two of the gauge groups. In addition there is a superpotential, which for a toric
Calabi--Yau is always a sum of monomials, such that each matter field appears in exactly two terms with opposite signs. Figure
\ref{fig-quiver} shows the quiver and the superpotential corresponding to $dP_1$, where $dP_1$ denotes the toric Calabi--Yau that
is given by a complex cone over the first del Pezzo surface.

\begin{figure}
\begin{center}
\includegraphics[scale=0.3]{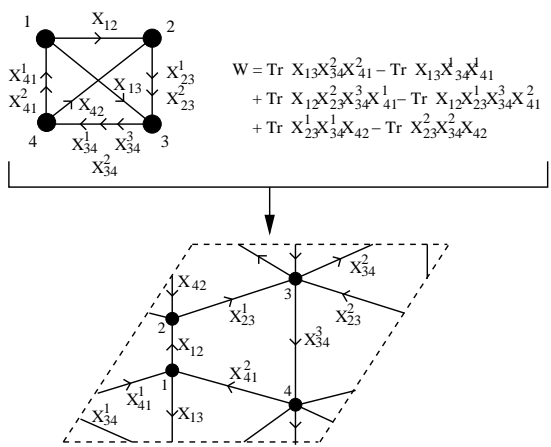}
\end{center}
\caption{\small{Top : The quiver and the superpotential for $dP_1$. Bottom : The corresponding planar quiver.}}
\label{fig-quiver}
\end{figure}

It is well known that any such quiver can be `opened up' and placed on a torus in such a way that the faces of the new graph
correspond to the superpotential terms; a clockwise orientation of a face corresponds to a negative term and vice versa
\cite{dimer2}. This is known as the planar (or periodic) quiver, and is also pictured in figure \ref{fig-quiver}. We will mostly
work with the universal covering of the planar quiver, which is an infinite, periodic graph on $\mathbb{R}^2$. We denote the
quiver by $Q = (Q_0,Q_1,Q_2)$, where $Q_2$ contains the faces, i.e. superpotential terms.

\paragraph{The brane tiling and perfect matchings:} The dual graph of the planar quiver is particularly useful. It is generated by replacing the faces of
the planar quiver by vertices, replacing vertices by faces, and replacing the arrows with perpendicular edges.  See figure
\ref{fig-tiling} for an illustration for the case of $dP_1$. Since the faces of the planar quiver correspond to either positive
or negative superpotential terms, we can color the vertices of the brane tiling to reflect this sign; we choose black vertices
for negative terms and white for positive ones. Thus we see that the brane tiling is a bipartite graph; white vertices are only
connected to black vertices and vice versa. Therefore, a brane tiling $G = (G_0^+,G_0^-,G_1,G_2)$ \footnote{Here I follow the
notation of \cite{reineke}.} consists of two sets of vertices, $G_0^+$ and $G_0^-$, corresponding to positive and negative
superpotential terms, a set of edges $G_1 \subset G_0^+ \times G_0^-$ corresponding to the bifundamental matter fields, and a set
of faces corresponding to the gauge groups.

\begin{figure}
\begin{center}
\includegraphics[scale=0.4]{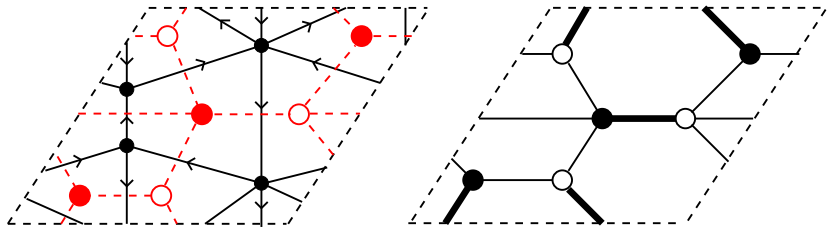}
\end{center}
\caption{\small{Left : The planar quiver (in black) and its dual graph, the brane tiling (in red). Right : The brane tiling. The
thickened lines make up one of the eight perfect matchings of $dP_1$. }} \label{fig-tiling}
\end{figure}

A perfect matching of a tiling $G$ is a subset of the edges such that each vertex in $G_0$ is touched by exactly one edge in the
perfect matching. Figure \ref{fig-tiling} shows one of the perfect matchings of $dP_1$.  A given bipartite graph need not have a
perfect matching; however, I will show in section 3 that for a physical brane tiling, that is a bipartite graph arising from a
configuration of D-branes, a perfect matching always exists.

\paragraph{The isoradial embedding and rhombus lattice:} To prove the consistency conditions of the crystal model
\cite{ooguri,reineke} we need to take into account the fact that the low energy field theory on the world-volume of the branes is
conformal. This means that the NSVZ beta functions for all gauge groups must vanish, which implies (see \cite{novikov,kennaway})
\begin{equation}
\sum_{i\in a} (1-R(X_i)) = 2, \quad \textrm{for all gauge groups } a, \Label{beta}
\end{equation}
\begin{figure}
\begin{center}
\includegraphics[scale=0.4]{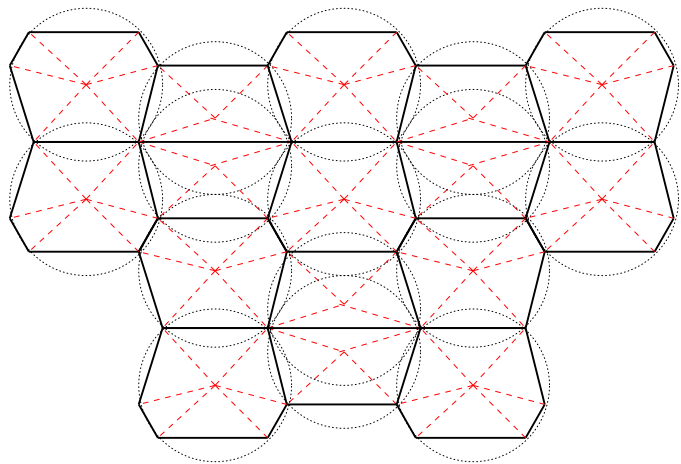}
\end{center}
\caption{\small{The isoradial embedding. Dotted black lines : unit circles. Solid black lines : the brane tiling (several unit
cells; the black and white vertices have been suppressed for clarity). Dashed red lines : Rhombus lattice. }} \label{fig-isorad}
\end{figure}
where the sum is over all the matter fields transforming under that gauge group $a$, and $R(X_i)$ denotes the R-charge of the
corresponding field. Further, the superpotential has to be marginal, so each term in the superpotential must have R-charge 2.
Recalling that in the planar quiver the superpotential terms correspond to faces, this implies
\begin{equation}
\sum_{i\in F} R(X_i) = 2, \quad \textrm{for all superpotential terms } F, \Label{W}
\end{equation}
where $F$ is a face in the planar quiver, and the sum is over all the edges surrounding that face. The relations (\ref{beta}) and
(\ref{W}) can be written in terms of the brane tiling as
\begin{eqnarray}
& & \sum_{e\in V} R_e = 2, \Label{W-2}\\
& & \sum_{e\in F} R_e = \textrm{\#edges } - 2, \Label{beta-2}
\end{eqnarray}
where $V$ is any vertex on the brane tiling and the first sum is over all edges connected to that vertex, and $F$ is a face on
the tiling and the second sum is over all the edges surrounding that face.

\begin{figure}
\begin{center}
\includegraphics[scale=0.4]{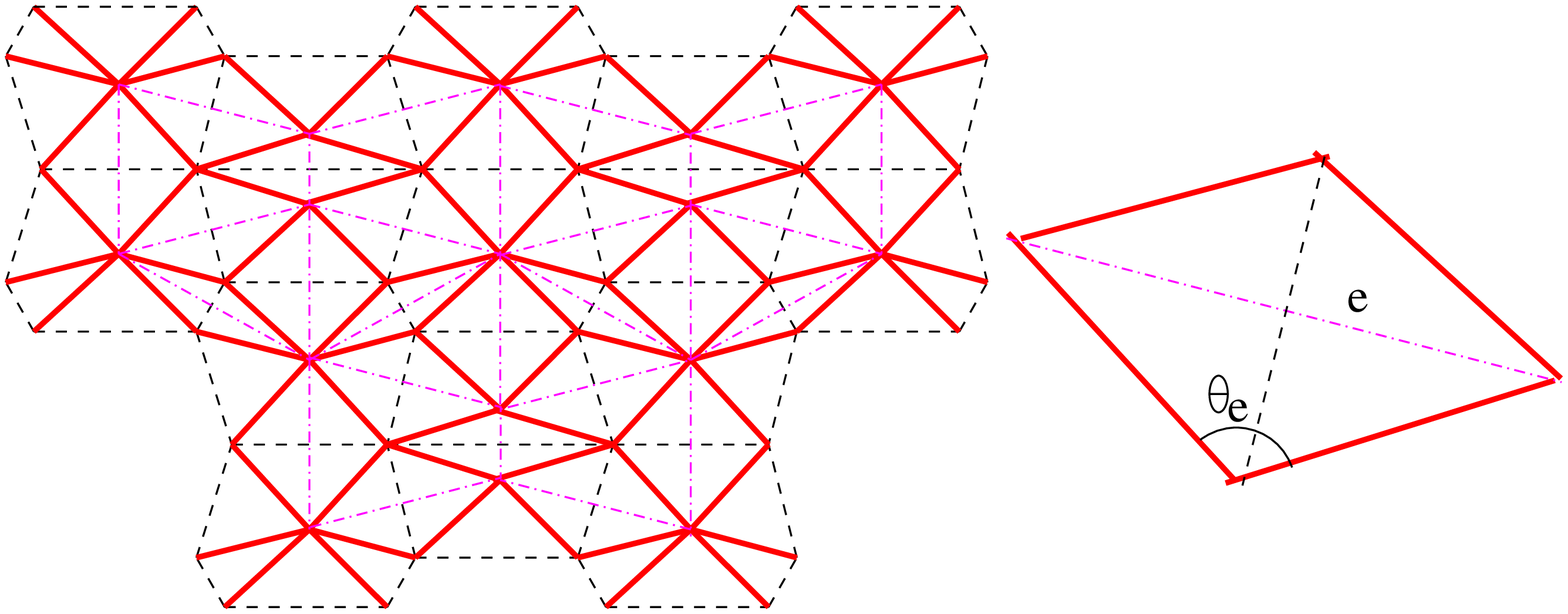}
\end{center}
\caption{\small{Left : the rhombus lattice. The dashed black diagonals give the brane tiling, and the dash-dotted magenta
diagonals give the planar quiver. Right : One of the rhombi characterized by the angle $\theta$. }} \label{fig-isorad2}
\end{figure}

These relations have a very nice geometric interpretation. Consider an embedding of the brane tiling where one draws a unit
circle for each of the faces of the tiling, and places the vertices on the circumference of the circle. This is known as the
isoradial embedding, and is shown in figure \ref{fig-isorad} for $dP_1$. Next draw lines from the centers of the circles to the
vertices on the circumference; this yields a lattice consisting of rhombi and is shown in figures \ref{fig-isorad} and
\ref{fig-isorad2}. Half of the vertices of the rhombi are on the vertices of the brane tiling, i.e. the superpotential terms,
while the other half of the vertices are on the faces of the tiling, i.e. the gauge groups. Thus, depending on which diagonals we
focus on we get the brane tiling or the planar quiver; the rhombus lattice contains the data of both graphs\footnote{We have
suppressed the black and white vertices of the brane tiling, or equivalently the direction of the arrows of the quiver. This data
needs to be included in the rhombus lattice for it to contain the same data as the tiling or quiver.}.

\begin{figure}
\begin{center}
\includegraphics[scale=0.3]{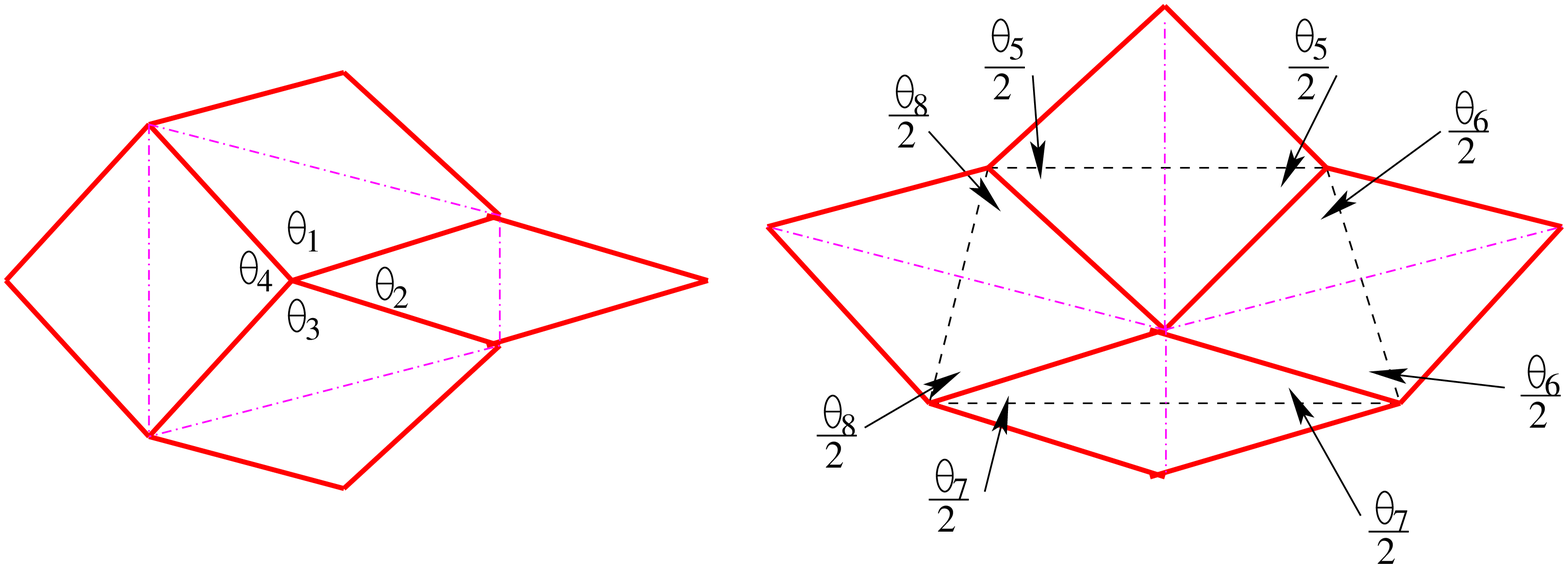}
\end{center}
\caption{\small{Left : Rhombi sharing a superpotential vertex. Right : Rhombi sharing a gauge group vertex. }}
\label{fig-isorad3}
\end{figure}

Each rhombus is characterized by an angle $\theta_e$, which we choose to be opposite the diagonal corresponding to an arrow in
the quiver, and we associate the angle $\theta_e$ to the R-charge $R_e$ of the corresponding field by
\begin{equation}
R_e \equiv \frac{\theta_e}{\pi},
\end{equation}
we see that (\ref{W-2}) and (\ref{beta-2}) are automatically satisfied; the first equation translates to the angles summing to
$2\pi$ around a superpotential vertex, while the second translates to the sum of internal angles in a polygon being (\#edges
-2)$\pi$. This is shown in figure \ref{fig-isorad3}. Note that for the rhombus lattice to be non-degenerate we need all R-charges
to lie in the interval $(0,1)$; this holds for any physical tiling \cite{rhombus}.

\paragraph{Train tracks and zig-zag paths:} The rhombus lattice has a set of special paths, known as rhombus paths or train tracks. These are constructed
by joining together rhombi sharing parallel edges, as in figure \ref{fig-zigzag}. Each such path defines a path in the planar
quiver that turns alternately maximally left and right, known as a zig-zag path.  It is known that for a tiling arising from a
system of branes, i.e. one having an isoradial embedding, a zig-zag path never intersects itself, and two zig-zag paths intersect
at most at one edge \cite{kenyon1,kenyon2,rhombus}.
\paragraph{The path algebra:} Consider the algebra $\mathbb{C}Q$, which consists of all paths in the quiver. The multiplication
in $\mathbb{C}Q$ is defined by joining two paths when the first ends where the second begins, otherwise the product is zero.
Since the arrows on the quiver correspond to bifundamental fields, the F-term relations give relations in the path algebra.
Figure \ref{fig-paths} shows three paths in the path algebra of $dP_1$ which are seen to be identical by using the F-terms
relations. If we denote the ideal generated by the F-term conditions by $\mathcal{F}$, the independent paths are given by the
factor algebra $A = \mathbb{C}Q / \mathcal{F}$.

\begin{figure}
\begin{center}
\includegraphics[scale=0.2]{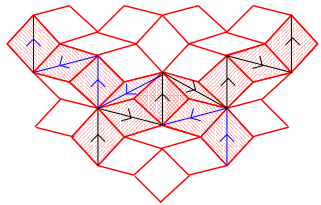}
\end{center}
\caption{\small{Two rhombus paths intersecting. Black and blue lines denote the zig-zag paths in the quiver corresponding to the
rhombus paths.}} \label{fig-zigzag}
\end{figure}

Let us denote the shortest\footnote{`Shortest' meaning of smallest R-charge.} path between vertices $i$ and $j$ by $v_{ij}$.
Also, denote the loop around a face in the quiver by $w$; all such loops are equivalent by the F-term relations. It was shown in
\cite{reineke} that the elements in the factor algebra $A$ can be written in the form $v_{ij} w^n$, where $n\ge 0$ is a natural
number.

\section{Proofs of the consistency conditions}
The crystal model \cite{ooguri} is built on the mathematical work \cite{reineke}, who place three consistency conditions for the
brane tiling: (3.5), (4.12) and (5.3) in \cite{reineke}. The first two of these are fairly trivial, though we prove (3.5) in
detail since as a consequence we can also prove that any physical brane tiling has a perfect matching. Although perfect matchings
are widely used in the physics literature, to my best knowledge this is the first direct proof\footnote{The Kasteleyn matrix and
its determinant, the characteristic polynomial $\sum_{i,j} c_{ij}z^iw^j$, are central to the dimer literature, and perfect
matchings appear in the coefficients $c_{ij}$. If there were no perfect matchings the coefficients would vanish. In \cite{amoeba}
it was argued via mirror symmetry that these coefficients are generally nonzero; the proof given here verifies this result
directly without requiring mirror symmetry.} that they exist for any physical tiling. Most of this section is devoted proving
condition (5.3) of \cite{reineke}.

\paragraph{Condition 3.5:} The first consistency condition of \cite{reineke} is that the tiling should be non-degenerate, meaning
that every edge in the brane tiling belongs to some perfect matching. We will proceed to show that this holds for physical
tilings using the following lemma:

\begin{figure}
\begin{center}
\includegraphics[scale=0.2]{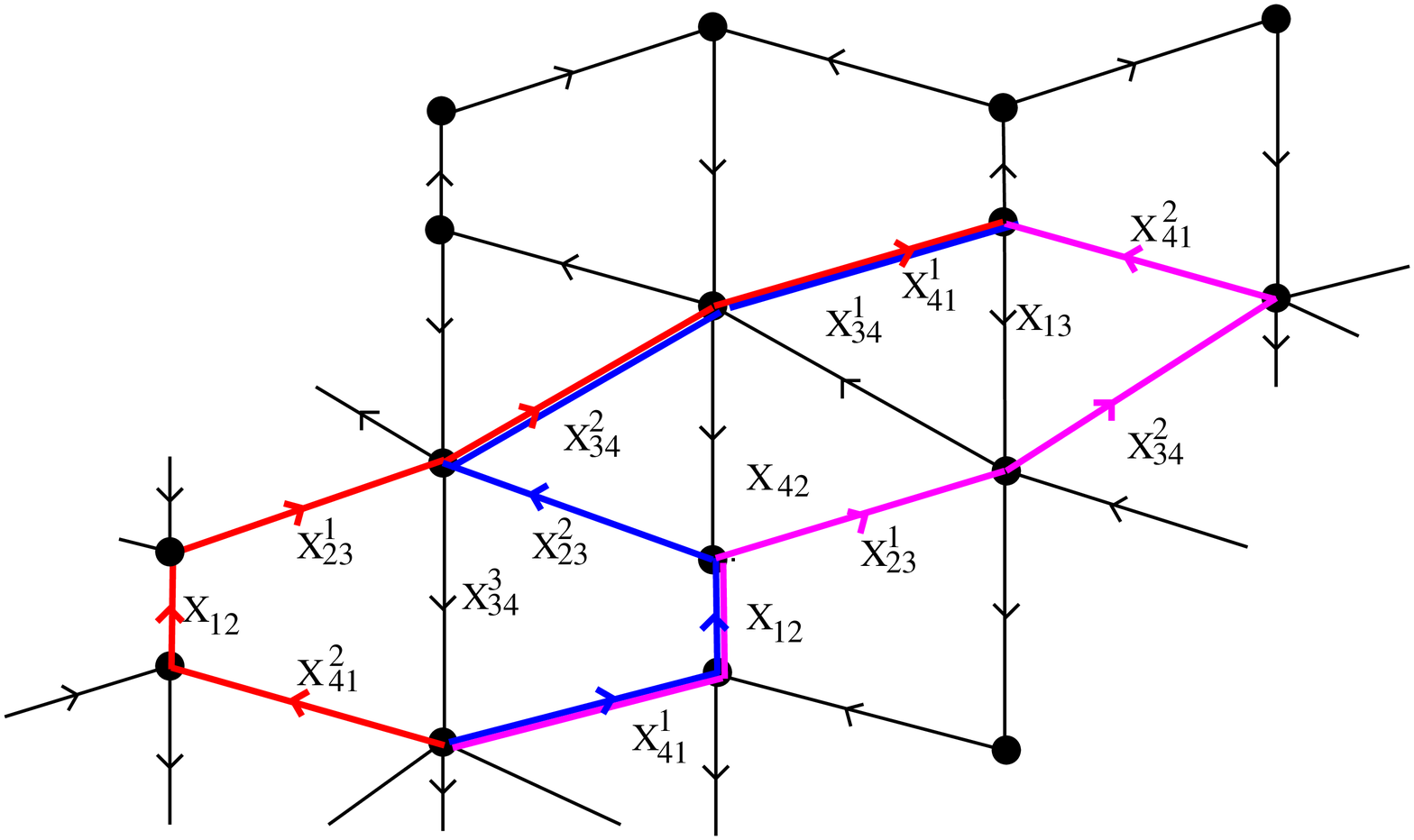}
\end{center}
\caption{\small{Three equivalent paths in the planar quiver of $dP_1$. To see the equivalence of the red and blue paths, use
$\partial W /
\partial X_{34}^2=0$, and to show that the magenta path is equivalent to these two use $\partial W /
\partial X_{13}=0$ and $\partial W /
\partial X_{42}=0$. (Refer to figure \ref{fig-quiver} for the superpotential.)}} \label{fig-paths}
\end{figure}

\begin{lem}
Any physical brane tiling $G = (G_0^+,G_0^-,G_1,G_2)$ has a perfect matching.
\end{lem}
\begin{proof}[Proof:] For any subset $A \subset G_0^+$ of white vertices, we define the set of neighbors $N(A) \subset G_0^-$ to consist
of all the black vertices that are connected to vertices in $A$ by an edge in $G_1$. Then Hall's theorem\footnote{Also known as
the `marriage theorem'.} states that
\begin{quote}
A bipartite graph $G$ has a perfect matching if and only if $|A| \le |N(A)|$ for any subset $A\subset G_0^+$.
\end{quote}
Now choose a subset $A \subset G_0^+$. We need to show $|N(A)| \ge |A|$. This follows from marginality of the superpotential;
since each superpotential term has R-charge 2, the edges connected to the vertices in $A$ have total R-charge $2|A|$. Since each
of the edges is connected to a vertex in $N(A)$, this charge is divided among the neighboring vertices $N(A)$. Since each vertex
can have only two units of R-charge, the pigeonhole principle states that $|N(A)| \ge |A|$. This is illustrated in figure
\ref{fig-bad}. This proves that $G$ has a perfect matching.
\end{proof}

\begin{figure}
\begin{center}
\includegraphics[scale=0.2]{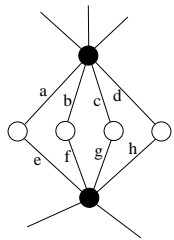}
\end{center}
\caption{\small{Part of an inconsistent tiling: since the white vertices have R-charge 2, we have $a+e=b+f=c+g=d+h=2$, which
implies one of the black vertices must have R-charge greater than two, which conflicts with the marginality of $W$. Similar
arguments were used in a different setting recently in \cite{circuit}.}} \label{fig-bad}
\end{figure}

We can go a bit further and show that $|A| < |N(A)|$ when $A \neq \emptyset, G_0^+$. First assume $|A| = |N(A)|$. Then by the
previous argument each vertex in $|N(A)|$ must already have R-charge 2, and therefore can't be connected to any new vertices,
i.e. $N(N(A)) = A$. But this implies that either $A=G_0^+$, $A=\emptyset$,  or the tiling consists of two or more disconnected
pieces, which is impossible. Thus either $A=\emptyset$ or $A=G_0^+$.

\begin{thm}
Physical brane tilings are non-degenerate, i.e. each edge belongs to a perfect matching.
\end{thm}

\begin{proof}[Proof:] Pick an edge $e \in G_1$. Let us construct a new bipartite graph $G'$ by removing from $G$ the two vertices that are
connected by $e$, and all the edges that are connected to either of those vertices. Note that if $G$ only had two vertices, each
edge is already a perfect matching; we can thus assume $G'\neq \emptyset$. Pick any subset of white vertices $A\subset G_0^{'+}$.
Then $G'$ has a perfect matching iff $|N'(A)| \ge |A|$.\footnote{$N'(A)$ denotes the neighbors in the tiling $G'$, and $N(A)$
denotes the neighbors in the tiling $G$.} This follows since $G'$ is a subtiling of $G$: $A$ is also a subset of white vertices
in $G$. Furthermore, since $A \neq G_0^+$ we have $|A| < |N(A)|$, where the neighbors are taken in $G$. Restricting to $G'$ we
see that $|N'(A)| = |N(A)|$ or $|N'(A)| = |N(A)|-1$, depending on whether the removed vertex was a neighbor to $A$. Thus $|A| \le
|N'(A)|$, and therefore by Hall's theorem $G'$ has a perfect matching $M'$. This proves the theorem, since we can define $M = M'
\cup e$, which is a perfect matching of the tiling $G$, and contains the edge $e$.
\end{proof}

\paragraph{Condition 5.3:} Having proved condition 3.5, we turn our attention to the second consistency condition:

\begin{thm}
For any two vertices $i,j$ in the planar quiver, there exists an arrow $a : j\to k$ such that $v_{ik} \sim a v_{ij}$.
\llabel{thm-main}
\end{thm}
To prove this, we will need a few useful lemmas:
\begin{lem}
If two vertices $i$ and $j$ are along the same zig-zag path, the zig-zag path provides $v_{ij}$, the shortest path between $i$
and $j$. \llabel{lem-zigminimum}
\end{lem}

\begin{figure}
\begin{center}
\includegraphics[scale=0.2]{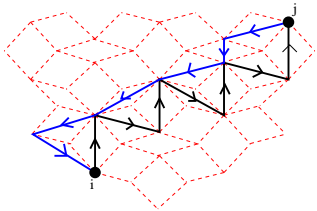}
\end{center}
\caption{\small{A zig-zag path from $i$ to $j$ (black), and the `neighboring' path back to $i$ (blue). }} \label{fig-zig-min}
\end{figure}

\begin{proof}[Proof:] Consider the path from $i$ to $j$ and back again to $i$ portrayed in figure \ref{fig-zig-min}.  This path is
equivalent to $w^n$, where $n$ is the number of faces circled by the path. Due to the structure of the rhombus lattice, $n$ is
the minimal number of faces that will be circled by any path $i\to j \to i$, and thus the path is equivalent to $v_{iji}$, the
path of minimal length from $i$ to $i$ passing through $j$. But clearly $v_{iji} \sim v_{ji}v_{ij}$, which proves that the
zig-zag path is equivalent to $v_{ij}$.
\end{proof}

\begin{lem}
For any two vertices $i$  and $j$, if there is a zig-zag path $Z_1$ passing through $j$ such that $v_{ik} \sim av_{ij}$ and
$v_{il} \sim bav_{ij}$, then $v_{im} \sim cbav_{ij}$.  Here $a$,$b$ and $c$ are the next three arrows in the zig-zag path
starting from $j$, and $k$,$l$ and $m$ are the end points of $a$,$b$ and $c$; see figure \ref{fig-main}. \llabel{lem-tough}
\end{lem}

\begin{figure}
\begin{center}
\includegraphics[scale=0.35]{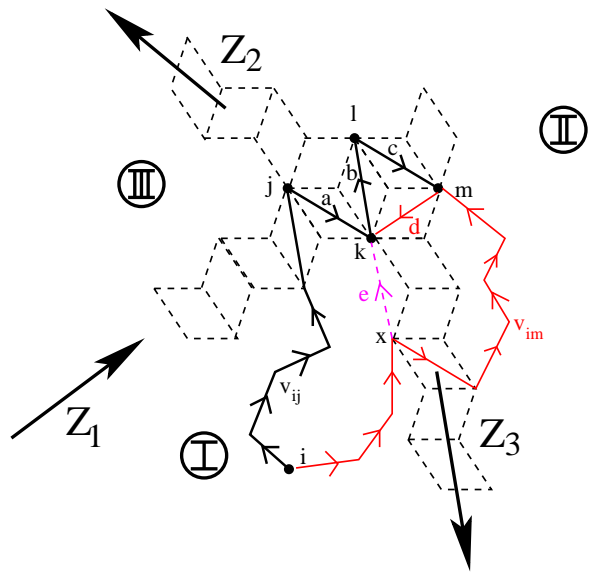}
\end{center}
\caption{\small{The set-up for lemma \ref{lem-tough}. (This is for a generic tiling, not $dP_1$.)}} \label{fig-main}
\end{figure}

\begin{proof}[Proof:] Each arrow in the quiver belongs to exactly two zig-zag paths. Thus, there are unique zig-zag paths different from
$Z_1$ that pass through $b$ and $c$, let us denote these by $Z_2$ and $Z_3$ respectively. These are also shown in figure
\ref{fig-main}; however the reader should note that one shouldn't rely too much on the figure: the way the figure has been drawn
it appears that $Z_2$ and $Z_3$ intersect in the rhombus whose diagonal arrow is denoted $d$. This need not be the case; $d$ is
meant to represent the path that completes the face on which the vertices $k,l,m$ lie; if the face is a triangle as drawn then
$Z_2$ and $Z_3$ intersect there as shown in the figure, if the face is a polygon with more than three edges, then $d$ is not an
edge but rather a path of edges completing the face, in which case $Z_2$ and $Z_3$ do not intersect in that rhombus. The reader
can verify that the proof presented here does not rely on such aspects of the graph.

The three zig-zag paths divide the plane into regions denoted I, II and III in the figure. I have drawn vertex $i$ to reside in
region I; the reader can verify that if $i$ was in region II the conditions of the lemma would not be satisfied, and if it
resided in region III an argument similar to the one given below will show the lemma to be true.

Now let us assume the lemma is not true: $cba v_{ij} \sim w v_{im}$, i.e. there exists a path $v_{im} : i \to m$ that is shorter
than $cbav_{ij}$. This path is shown in figure \ref{fig-main} as a red path. Denote by $x$ the vertex where $v_{im}$ hits the
zig-zag path $Z_3$.\footnote{We have drawn $v_{im}$ to cross $Z_3$; if this is not the case, i.e. the path $v_{im}$ is directed
towards left in the figure, then it must cross the original zig-zag path $Z_1$ and a similar argument applies. Thus we can assume
$v_{im}$ crosses $Z_3$.} Now, we can take the path $wv_{im}$ to be given by travelling first along $v_{im}$ from $i$ to $m$, and
then around the face given by $cbd$.\footnote{Recall the previous comment that $d$ need not be a single arrow, it represents the
rest of the arrows around this face.} Then we have $cba v_{ij} \sim cbdv_{im}$, and we can cancel $cb$ to yield $dv_{im} \sim
av_{ij} \sim v_{ik}$, i.e. $dv_{im}$ is the shortest path $i\to k$. Now consider path $V$, which is defined by travelling on
$v_{im}$ until $x$, and then taking the magenta path $e$ to $k$. Here $e$ is a neighboring path to $Z_3$ in the sense of the
proof of lemma \ref{lem-zigminimum}; it need not be a single arrow as drawn in the figure. It was shown in the proof of lemma
\ref{lem-zigminimum} that such a neighboring path is the shortest path between its endpoints. It is also easily shown that the
red path $x\to k$ is longer than the path $e$; thus we have constructed a path $V : i\to k$ that is shorter than $dv_{im}$. This
is a contradiction since our assumption that the lemma is false yielded $dv_{im} \sim v_{ik}$; thus our assumption must be wrong
and the lemma holds.

The readers can easily convince themselves that if $i$ was located in region $III$ or if the path $v_{im}$ crossed $Z_1$ instead
of $Z_3$ similar arguments can be used to prove the lemma.
\end{proof}

\begin{proof}[Proof of theorem \ref{thm-main}:] Let $i$ and $j$ be two vertices, and $v_{ij}$ the shortest path between them. We need to
show that there exists an arrow $a : j\to k$ such that $av_{ij} \sim v_{ik}$. Consider the last arrow on path $v_{ij}$, which we
denote by $b$. It belongs to exactly two zig-zag paths; one turning maximally right and one turning maximally left at $j$; let us
denote the next arrows on these paths by $a_R$ and $a_L$ respectively; see figure \ref{fig-main2}.

\begin{figure}
\begin{center}
\includegraphics[scale=0.3]{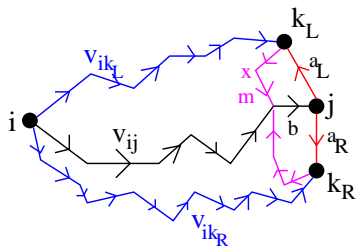}
\end{center}
\caption{\small{The set-up for theorem \ref{thm-main}.}} \label{fig-main2}
\end{figure}

We claim that either $a_L$ or $a_R$ can be identified with $a$ so that the theorem holds.  We prove this by assuming that one of
them does not have the property of $a$, and then showing that the other one must have it. Assume $a_L$ does not have the
property, i.e. $a_L v_{ij} \sim w v_{ik_L}$. We can take $wv_{ik_L}$ to be first $v_{ik_L} : i \to k_L$, and then a loop around
the face indicated in figure \ref{fig-main}, which is surrounded by $a_L$, $b$, and the magenta path $m$. Denote by $x$ the last
vertex before the magenta path hits $b$; if the face is a triangle then $x=k_L$. Since this path shares the same last arrow with
$a_Lv_{ij}$, we can cancel $a_L$ to yield $v_{ij} \sim bmv_{ik_L}$. Thus the path constructed is also the shortest path $i\to j$,
and trivially also the shortest path to any vertex it passes. Following this path to $x$ then gives the shortest path $v_{ix}$.
But now note that the indicated path from $x$ to $k_R$ is a zig-zag path: one first follows the last arrow of the magenta path,
turns maximally left, follows $b$ and then turns maximally right. But we also know that this constructed path is the shortest
path $i\to j$, so the conditions of lemma \ref{lem-tough} are satisfied, and if we take another step on the zig-zag path, by
lemma \ref{lem-tough} we again have a shortest path, and thus $a_R v_{ij} \sim v_{ik_R}$, and the theorem is proved.
\end{proof}

\section*{Acknowledgments}
I am very thankful to Mark Van Raamsdonk and Masahito Yamazaki for helpful comments during the preparation of this note. I am
supported in part by the Natural Sciences and Engineering Research Council of Canada and the Institute for Particle Physics.


\begin{thebibliography}{99}

\bbibitem{dimer1}
  A.~Hanany and K.~D.~Kennaway,
  ``Dimer models and toric diagrams,''
  arXiv:hep-th/0503149.

\bbibitem{dimer2}
  S.~Franco, A.~Hanany, K.~D.~Kennaway, D.~Vegh and B.~Wecht,
  ``Brane Dimers and Quiver Gauge Theories,''
  JHEP {\bf 0601}, 096 (2006)
  [arXiv:hep-th/0504110].

\bbibitem{dimer3}
  S.~Franco, A.~Hanany, D.~Martelli, J.~Sparks, D.~Vegh and B.~Wecht,
  ``Gauge theories from toric geometry and brane tilings,''
  JHEP {\bf 0601}, 128 (2006)
  [arXiv:hep-th/0505211].

\bbibitem{rhombus}
  A.~Hanany and D.~Vegh,
  ``Quivers, tilings, branes and rhombi,''
  JHEP {\bf 0710}, 029 (2007)
  [arXiv:hep-th/0511063].

\bbibitem{amoeba}
  B.~Feng, Y.~H.~He, K.~D.~Kennaway and C.~Vafa,
  ``Dimer models from mirror symmetry and quivering amoebae,''
  Adv.\ Theor.\ Math.\ Phys.\  {\bf 12}, 3 (2008)
  [arXiv:hep-th/0511287].

\bbibitem{kont}
  M.~Kontsevich,
  ``Homological Algebra of Mirror Symmetry,'' in ``Proceedings of the International Congress of Mathematicians'', 120-139,
  Birkh\"auser, 1995. [alg-geom/9411018].

\bbibitem{douglas1}
  M.~R.~Douglas,
  ``D-branes, categories and N = 1 supersymmetry,''
  J.\ Math.\ Phys.\  {\bf 42}, 2818 (2001)
  [arXiv:hep-th/0011017].

\bbibitem{asp-zero}
  P.~S.~Aspinwall and A.~E.~Lawrence,
  ``Derived categories and zero-brane stability,''
  JHEP {\bf 0108}, 004 (2001)
  [arXiv:hep-th/0104147].

\bbibitem{asp-review}
  P.~S.~Aspinwall,
  ``D-branes on Calabi-Yau manifolds,''
  arXiv:hep-th/0403166.

\bbibitem{asp-exp}
  P.~S.~Aspinwall and I.~V.~Melnikov,
  ``D-branes on vanishing del Pezzo surfaces,''
  JHEP {\bf 0412}, 042 (2004)
  [arXiv:hep-th/0405134].

\bbibitem{her-exp}
  C.~P.~Herzog and R.~L.~Karp,
  ``Exceptional collections and D-branes probing toric singularities,''
  JHEP {\bf 0602}, 061 (2006)
  [arXiv:hep-th/0507175].

\bibitem{han-exp}
  A.~Hanany, C.~P.~Herzog and D.~Vegh,
  ``Brane tilings and exceptional collections,''
  JHEP {\bf 0607}, 001 (2006)
  [arXiv:hep-th/0602041].

\bibitem{asp-tor}
  P.~S.~Aspinwall,
  ``D-Branes on Toric Calabi-Yau Varieties,''
  arXiv:0806.2612 [hep-th].

\bbibitem{ooguri}
  H.~Ooguri and M.~Yamazaki,
  ``Crystal Melting and Toric Calabi-Yau Manifolds,''
  arXiv:0811.2801 [hep-th].

\bbibitem{reineke}
  S.~Mozgovoy and M.~Reineke,
  ``On the noncommutative Donaldson-Thomas invariants arising from brane
  tilings,''
  arXiv:0809.0117 [math.AG].

\bbibitem{kennaway}
  K.~D.~Kennaway,
  ``Brane Tilings,''
  Int.\ J.\ Mod.\ Phys.\  A {\bf 22}, 2977 (2007)
  [arXiv:0706.1660 [hep-th]].

\bbibitem{yamazaki}
  M.~Yamazaki,
  ``Brane Tilings and Their Applications,''
  Fortsch.\ Phys.\  {\bf 56}, 555 (2008)
  [arXiv:0803.4474 [hep-th]].

\bbibitem{novikov}
  V.~A.~Novikov, M.~A.~Shifman, A.~I.~Vainshtein and V.~I.~Zakharov,
  ``Exact Gell-Mann-Low Function Of Supersymmetric Yang-Mills Theories From
  Instanton Calculus,''
  Nucl.\ Phys.\  B {\bf 229}, 381 (1983).

\bbibitem{kenyon1}
 R.~Kenyon,
 ``An introduction to the dimer model,''
 arXiv:math.CO/0310326.

\bbibitem{kenyon2}
 R.~Kenyon and J.~M.~Schlenker,
 ``Rhombic embeddings of planar graphs with faces of degree 4,''
 arXiv:math-ph/0305057.

\bbibitem{circuit}
  V.~Balasubramanian, B.~Czech, A.~D.~Shapere and B.~Wecht,
  ``Quiver Topology and RG Dynamics,''
  arXiv:0811.4427 [hep-th].

\end{thebibliography}
\end{document}